\documentclass[english]{lipics-v2021}

\nolinenumbers

\usepackage{booktabs}
\usepackage{mathtools}
\newcommand{\Ocal}{\mathcal{O}}
\newcommand{\Qdet}{Q_{\mathrm{det}}^*}
\newcommand{\Qrand}{Q_{\mathrm{rand},\delta}^*}
\newcommand{\Qround}[1]{Q_{\mathrm{det}}^{[#1]}}
\newcommand{\Qrandnonad}{Q_{\mathrm{rand,nonad},\delta}}
\newcommand{\Rexp}{R_{\delta}^{\mathrm{exp}}}
\newcommand{\Sub}{\operatorname{Sub}}
\newcommand{\SD}{\operatorname{sd}}
\newcommand{\bits}{\{0,1\}}
\newcommand{\eps}{\epsilon}

\title{Sharp Two-Round Adaptivity and Round Hierarchies for Semantic Regular Expressions}
\titlerunning{Sharp Adaptivity Hierarchies for SemREs}

\author{Runzhou Li}{The State Key Laboratory of Blockchain and Data Security, Zhejiang University}{darkpaper@zju.edu.cn}{}{}
\author{Hongfei Fu}{School of Computing and Artificial Intelligence, Shanghai University of Finance and Economics}{hongfeifu1984@gmail.com}{}{}
\author{Qingkai Shi}{The State Key Laboratory of Novel Software Technology, Nanjing University}{qingkaishi@nju.edu.cn}{}{}
\author{Peisen Yao}{The State Key Laboratory of Blockchain and Data Security, Zhejiang University}{payoaa@zju.edu.cn}{}{}

\authorrunning{R.\,Li, H.\,Fu, Q.\,Shi, and P.\,Yao}

\Copyright{Anonymous Author(s)}



\begin{document}

\maketitle

\begin{abstract}
Semantic regular expressions (SemREs) attach external Boolean predicates to
matched spans, making both the number and the sequentiality of oracle calls
central resources.  For a fixed expression and word, we represent membership
by a polynomial-size monotone span circuit and identify optimal semantic
evaluation with Boolean decision-tree evaluation.

We determine the extremal power of adaptivity asymptotically sharply.  For
every $E\ge2$, there is a unary, star-free, semantic-depth-one instance of
syntax size $\Theta(E)$ with $E$ essential oracle keys and only unit-length
semantic spans whose one-round cost is $E$, whereas its exact two-round and
unrestricted deterministic costs are
\[
 \log_2 E+\tfrac12\log_2\log_2 E+O(1).
\]
Consequently, the largest nonadaptive-to-adaptive ratio is
$(1+o(1))E/\log_2E$, including the optimal leading constant.  A second
restricted family exhibits a complete round hierarchy: its optimal $R$-round
cost is $\Theta(R E^{1/R})$.  Thus the maximal gap already appears in two
rounds, while other instances interpolate smoothly across all round budgets.

Both constructions admit one-predicate realizations over the fixed alphabet
$\{0,1,\#\}$ with logarithmic-length semantic spans and
$O(E\log^2 E)$ total representation size.  Under pointwise error $\delta<1/2$
and worst-case expected cost, randomized nonadaptive complexity is exactly
$(1-2\delta)E$ for every instance with $E$ essential keys.  Finally, for a
fixed word $w$ and $h$ predicate names, the exact randomized minimax value is
$(1-2\delta)h\SD(w)$, where $\SD(w)$ counts distinct substring values; a span
bound $s$ replaces $\SD(w)$ by $\SD_s(w)$.  These results separate semantic
information acquisition, parallel latency, and local symbolic matching cost.
\end{abstract}

\section{Introduction}
\label{sec:introduction}

Regular expressions describe syntactic structure.  Semantic regular
expressions (SemREs) additionally attach predicates to matched spans and
delegate those predicates to an external oracle.  A predicate may ask whether
a span names a city, expresses a positive opinion, or denotes a
security-relevant action.  The oracle may be a language model, a learned
classifier, a database service, or a human annotator.  In each case, an oracle
call can cost far more than an ordinary string-processing step.

This difference creates two resources absent from classical
regular-expression matching.  The first is the amount of \emph{semantic
information}: how many distinct predicate--substring values must be obtained?
The second is \emph{semantic latency}: how many sequential rounds of calls are
needed if independent queries can be issued in parallel?  An accepting parse
may be certified by a few positive answers, a negative instance may require
ruling out every candidate span, and an adaptive matcher may use early answers
to decide which later values matter.  Counting filter occurrences or span
positions alone does not capture these effects.

SemREs were introduced for synthesizing extraction patterns from examples
\cite{chen2023semantic}.  Huang et al.\ subsequently gave a concrete
SNFA/query-graph membership algorithm with an $O(|r||w|^2)$ query bound and
matching worst-case mechanisms in the relevant regimes
\cite{huang2025membership}.  We instead minimize over \emph{all} correct
matchers after both the expression $r$ and word $w$ have been fixed:
\begin{quote}
How many oracle values are intrinsically required by one SemRE membership
instance, and how does the optimum change with the number of adaptive rounds?
\end{quote}

We work in a black-box, extensional model.  A query is a pair $(q,x)$ of a
predicate name and a string value; repeated occurrences of the same pair share
one answer, while distinct pairs have no promised relationship.  Every fixed
instance therefore induces a finite monotone Boolean function $F_{r,w}$ of its
oracle values.  Our span-circuit theorem constructs this function in
polynomial size despite exponentially many parses, and the decision-tree
characterization makes Boolean query complexity an exact model of semantic
information acquisition.

\smallskip
\noindent\textbf{Sharp Maximal Gap in Two Rounds.}
Our headline result determines the largest possible adaptivity gap, including
its leading constant.  Let $e(r,w)$ denote the number of essential oracle
keys, and let $\Qround{R}(r,w)$ be the minimum worst-case number of queries made
by an exact deterministic matcher using at most $R$ parallel rounds.  Thus
$R=1$ is nonadaptive, while $\Qdet(r,w)$ allows unrestricted adaptivity.
For every integer $E\ge2$, Theorem~\ref{thm:sharp-adaptivity} constructs a
unary, star-free, semantic-depth-one instance $(r_E,w_E)$ such that
\[
 e(r_E,w_E)=E,
 \qquad
 |r_E|=\Theta(E),
 \qquad
 \Qround{1}(r_E,w_E)=E,
\]
all semantic spans have length one, and
\[
 \Qround{2}(r_E,w_E)=\Qdet(r_E,w_E)
 =\log_2E+\tfrac12\log_2\log_2E+O(1).
\]
Every decision tree with $E$ essential variables has depth at least
$\lceil\log_2(E+1)\rceil$.  Hence, if
\[
 \Gamma(E)=
 \max_{e(r,w)=E}\frac{\Qround{1}(r,w)}{\Qdet(r,w)},
\]
then
\[
 \Gamma(E)=(1+o(1))\frac{E}{\log_2E}.
\]
The previous $\Theta(E/\log E)$ statement left a factor of two in the
construction; the new theorem determines the asymptotically optimal constant
and shows that only two rounds are needed.

The Boolean starting point is the classical monotone addressing function,
which is known to combine shallow decision trees with exponentially many
relevant variables~\cite{kane2013monotone}.  We use a partial central-layer
version.  Querying $k$ address bits identifies at most one relevant data bit,
while an exact-degree argument proves depth $k+1$.  The SemRE contribution is
an exact calibration to every $E$, together with a linear-size realization
using a prefix-trie factorization.

\smallskip
\noindent\textbf{A Complete Round Hierarchy.}
The guarded recursive construction from the preliminary result has a different
strength.  If its height is $t$, it has $e_t=3\cdot2^t-2$ essential keys and
unrestricted cost $2t+1$.  Theorem~\ref{thm:round-hierarchy} determines its
cost for every $1\le R\le t$:
\[
 \Qround{R}(r_t,w_t)
 =\Theta\!\left(R\,2^{t/R}\right)
 =\Theta\!\left(R\,e_t^{1/R}\right).
\]
The upper bound queries complete blocks of the recursive tree.  The lower
bound restricts the guards to a promised pointer tree and uses an adversary
that preserves an untouched descendant subtree after each batch.  Rounds of
adaptivity are a standard refinement in query models
\cite{canonne2018adaptivity}; here they quantify batched semantic calls
directly.  With a cap of $B$ calls per RPC, the family needs
\[
 \Theta\!\left(
 \max\left\{1,\frac{\log e_t}{\log(B+1)}\right\}
 \right)
\]
rounds.  Thus monotone addressing realizes the maximal gap immediately,
whereas guarded recursion gives a smooth latency--work tradeoff.

\begin{table}[t]
\caption{The three extremal families play complementary roles.  Costs are
deterministic unless randomization is shown explicitly.}
\label{tab:result-map}
\small
\centering
\begin{tabular}{@{}llll@{}}
\toprule
Family & One round & $R$ rounds & Role \\
\midrule
Monotone address & $E$ & $\log E+\frac12\log\log E+O(1)$ for $R\ge2$
  & Sharp maximal gap \\
Guarded recursion & $\Theta(E)$ & $\Theta(RE^{1/R})$
  & Full round hierarchy \\
Factor OR & $E$ & $E$; randomized $(1-2\delta)E$
  & Exact minimax extreme \\
\bottomrule
\end{tabular}
\end{table}

\smallskip
\noindent\textbf{Fixed Vocabulary Without a Quadratic Blowup.}
The unary constructions use distinct predicate names to create distinct keys
for the same one-letter span.  Theorem~\ref{thm:sparse-dnf-compiler} gives a
generic packed encoding for a monotone DNF with $M$ terms and $L$ literal
occurrences.  Over the fixed alphabet $\{0,1,\#\}$ and one predicate name, it
uses only logarithmic-length filtered spans and total representation size
\[
 O((M+L)\log(e+1)).
\]
Applied to both hard families, this is $O(E\log^2E)$, improving the earlier
$O(E^2\log E)$ construction.  The logarithmic span length is optimal up to
constants for any fixed alphabet and one predicate.

\smallskip
\noindent\textbf{Exact Randomized Values.}
Under pointwise error $\delta<1/2$ and worst-case expected query cost,
Theorem~\ref{thm:random-nonadaptive} proves a general identity:
\[
 \Qrandnonad(r,w)=(1-2\delta)e(r,w).
\]
The upper bound mixes exact evaluation with the two constant outputs; the
matching lower bound couples two inputs witnessing each essential variable.
The same argument upgrades the fixed-word candidate-universe theorem.  If
$\SD(w)$ counts distinct substring values, then over all expressions using at
most $h$ predicate names,
\[
 \max_r \Qrand(r,w)=(1-2\delta)h\SD(w),
\]
and a semantic-span bound $s$ replaces $\SD(w)$ by $\SD_s(w)$.  For the
monotone-address family, sensitivity also gives randomized adaptive cost
$\Theta(\log E)$ for every fixed $\delta<1/2$, so the adaptivity separation
survives bounded error.

\smallskip
\noindent\textbf{Contribution Boundaries.}
The span circuit is an analysis interface, not a claim of a faster local
membership algorithm.  Our results separate symbolic parsing time, the number
of semantic values acquired, and sequential oracle latency.  The monotone
addressing Boolean idea and pointer-following lower-bound paradigm are known;
our claims concern their sharp calibration and restricted SemRE realizations,
the exact SemRE round tradeoff, the sparse fixed-vocabulary compiler, and the
randomized identities above.

\smallskip
\noindent\textbf{Organization.}
Section~\ref{sec:model} defines the oracle and round models, and
Section~\ref{sec:circuit} constructs the span circuit.
Section~\ref{sec:characterization} proves the decision-tree characterization
and the sharp two-round theorem.  Section~\ref{sec:round-hierarchy} establishes
the full round hierarchy, Section~\ref{sec:fixed-vocabulary} gives the packed
compiler, and Section~\ref{sec:lower-bounds} proves the randomized and
candidate-universe identities.  Section~\ref{sec:related} positions the
results, and Section~\ref{sec:conclusion} concludes.  Additional parameter
consequences and generic matching procedures appear in the appendices.

\section{Semantic Regular Expressions and Oracle Model}
\label{sec:model}

Let $\Sigma$ be a finite alphabet and let $\mathcal Q$ be a set of predicate
identifiers.  For a word $w=w_0\cdots w_{n-1}$ and indices
$0\le i\le j\le n$, write $w[i:j]=w_i\cdots w_{j-1}$; in particular,
$w[i:i]=\eps$.

\begin{definition}[Semantic regular expressions]
A \emph{semantic regular expression} over $\Sigma$ and $\mathcal Q$ is
generated by
\[
 r ::= \emptyset \mid \eps \mid a \mid r r \mid r\cup r \mid r^*
       \mid r\cap\langle q\rangle,
\]
where $a\in\Sigma$ and $q\in\mathcal Q$.  The first six constructs have their
usual regular-language semantics.  An oracle
$\Ocal:\mathcal Q\times\Sigma^*\to\bits$ interprets a semantic filter by
\[
 L(r\cap\langle q\rangle,\Ocal)
 =\{x\in L(r,\Ocal):\Ocal(q,x)=1\}.
\]
Thus the filter tests the entire string matched by its child expression.  The
membership problem asks whether $w\in L(r,\Ocal)$.
\end{definition}

For example, let $A=\bigcup_{a\in\Sigma}a$ and let $q_{\rm city}$ denote a
city-name predicate.  The expression
\[
 A^*(A^+\cap\langle q_{\rm city}\rangle)A^*
\]
accepts $w$ exactly when some nonempty factor $w[i:j]$ satisfies
$\Ocal(q_{\rm city},w[i:j])=1$.  The expression determines which spans are
possible; a matching algorithm determines which of their keys to query and in
what order.

The \emph{semantic depth} of an expression is the maximum number of filter
nodes on a root-to-leaf path.  We write $m=|r|$ and $n=|w|$.  Unless a bit-level
encoding is stated explicitly, $|r|$ is syntax-tree size: alphabet symbols and
predicate identifiers each count as one label.  Under a binary encoding, a
family using $e$ distinct identifiers may require an additional $O(\log e)$
bits per identifier; the fixed-vocabulary construction in
Section~\ref{sec:fixed-vocabulary} avoids this issue.

Because the grammar has no semantic negation, membership is monotone in the
oracle answers.  We omit complement, ordinary regular-expression intersection,
lookaround, and backreferences so that their succinctness costs do not enter
the semantic-query analysis.  Semantic negation would replace the monotone
span circuits below by general Boolean circuits.

\subsection{Atomic Oracle Access}

An \emph{atomic query} supplies one key $(q,x)$ and returns $\Ocal(q,x)$.  The
model has three assumptions.
\begin{enumerate}
  \item \emph{Extensionality.}  The answer depends on the predicate name and
  string value, not on the occurrence or coordinates of $x$ in $w$.
  \item \emph{Stability.}  Repeated queries to the same key return the same bit.
  \item \emph{No cross-key promises.}  Distinct keys have no promised
  relationship; a worst-case oracle may assign their bits independently.
\end{enumerate}

A deterministic matcher chooses its next query and final answer from $(r,w)$
and the query transcript.  It is \emph{adaptive} if later keys may depend on
earlier answers and \emph{nonadaptive} if all queried keys are fixed in
advance.  More generally, an \emph{$R$-round matcher} issues one parallel batch
in each of $R$ rounds; the batch in round $j$ may depend only on answers from
rounds $1,\ldots,j-1$.  Thus one round is nonadaptive.  We write
$\Qround{R}(r,w)$ for the minimum worst-case number of distinct queries made by
an exact deterministic $R$-round matcher and $\Qdet(r,w)$ when the number of
rounds is unrestricted.

An exact matcher must be correct for every oracle.  A randomized matcher with
error $\delta<1/2$ must, for every fixed oracle, return the correct answer with
probability at least $1-\delta$ over its private coins.  Its cost on that
oracle is the expected number of distinct keys queried, and its reported cost
is the maximum of this expectation over all oracles.  We write $\Qrand(r,w)$
for unrestricted randomized adaptivity and $\Qrandnonad(r,w)$ when the random
query set must be chosen before any answer is observed.

Stability makes memoization lossless, so each key is charged at most once.
Extensionality makes equal substring values at different positions share a
key.  Position- or context-sensitive predicates require coordinates or context
inside the key and therefore need not satisfy the distinct-value bounds below.
A matcher may request a key unrelated to membership, but
Theorem~\ref{thm:decision-tree} shows that such requests can be removed without
increasing deterministic or randomized cost.

Our primary work measure is the number of distinct atomic evaluations; round
complexity measures their sequentiality.  Local symbolic work and nonuniform
key costs are separate resources.  Appendix~\ref{sec:algorithms} returns to
those measures after the query-complexity results have been established.

\section{A Monotone Span-Circuit Representation}
\label{sec:circuit}

A SemRE may have exponentially many parses, so an analysis that explicitly
enumerates accepting parses is unsuitable even before oracle costs are
considered.  We instead collect all interval matches in one dynamic program.
The resulting circuit will serve as the interface between SemRE syntax and
Boolean query complexity.

Fix an expression $r$ and a word $w=w_0\cdots w_{n-1}$.  For every syntax node
$v$ of $r$ and interval $0\le i\le j\le n$, introduce a gate $M_v(i,j)$ with
the intended meaning
\[
 M_v(i,j)=1
 \quad\Longleftrightarrow\quad
 w[i:j]\in L(v,\Ocal).
\]
Equal substring values must receive the same variable even when they occur at
different positions.  We therefore preprocess $w$ with canonical substring
identifiers: two intervals receive the same identifier exactly when they spell
the same string.  For example, after constructing a suffix array and its LCP
array, one scan of the suffix order for each length groups equal factors in
$O(n^2)$ total time and space.  Circuit labels store these identifiers rather
than copying the strings.  For each distinct key, introduce an input variable
$X_{q,x}=\Ocal(q,x)$.  The leaf, union, concatenation, and filter gates are
defined as follows:
\begin{align*}
 M_{\emptyset}(i,j)&=0, &
 M_{\eps}(i,j)&=[i=j],\\
 M_a(i,j)&=[j=i+1\land w_i=a],\\
 M_{u\cup v}(i,j)&=M_u(i,j)\lor M_v(i,j),\\
 M_{uv}(i,j)&=\bigvee_{k=i}^{j}
       \bigl(M_u(i,k)\land M_v(k,j)\bigr),\\
 M_{u\cap\langle q\rangle}(i,j)&=
       M_u(i,j)\land X_{q,w[i:j]}.
\end{align*}
Here $[P]$ denotes the Boolean value of proposition $P$.

For Kleene star, set $M_{u^*}(i,i)=1$.  If $i<j$, define
\[
 M_{u^*}(i,j)=\bigvee_{k=i+1}^{j}
   \bigl(M_u(i,k)\land M_{u^*}(k,j)\bigr).
\]
The strict inequality $k>i$ makes this recurrence acyclic.  It does not lose
any accepting factorization: empty repetitions can be deleted from a
factorization of a nonempty word.  The output gate is $M_r(0,n)$.

We use a word-RAM representation in which interval endpoints, syntax-node
indices, predicate identifiers, and canonical substring identifiers occupy
$O(1)$ words.  We count bounded-fan-in Boolean gates.  Thus an $O(n)$-way
disjunction is expanded into $O(n)$ gates; equivalently, the same asymptotic
bound counts wires in an unbounded-fan-in presentation.

\begin{theorem}[Span-circuit theorem]
\label{thm:span-circuit}
For every expression $r$ and word $w$, the recurrences above construct a
monotone Boolean circuit $F_{r,w}=M_r(0,n)$ such that
\[
 F_{r,w}(\Ocal)=1
 \quad\Longleftrightarrow\quad
 w\in L(r,\Ocal).
\]
The circuit has size $O(mn^3)$ and can be evaluated in $O(mn^3)$ time and
$O(mn^2)$ space.
\end{theorem}

\begin{proof}
We prove the equivalence by structural induction on the syntax of $r$, using
increasing interval length for a star node.  The equations for leaves, union,
concatenation, and filtering follow directly from the language semantics.  For
star, every nonempty word in $L(u^*)$ has a first nonempty factor ending at
some $k>i$, followed by a suffix in $L(u^*)$; conversely, every disjunct gives
such a factorization.  This establishes correctness of the output gate.

Canonical substring identifiers are computed in $O(n^2)$ preprocessing time
and space as described above.  There are $O(mn^2)$ pairs consisting of a
syntax node and an input interval.
A concatenation or star entry examines $O(n)$ split points, whereas every other
entry has constant fan-in.  The circuit size and construction time are
therefore $O(mn^3)$.  Retaining one value for each syntax-node/interval pair
uses $O(mn^2)$ space.
\end{proof}

For certificate arguments, it is useful to unfold the same circuit
conceptually.  Nullable subexpressions make parse-based descriptions
ambiguous, so we use the acyclic recurrence itself.  A \emph{symbolic accepting
proof} chooses one child at each OR gate and retains both children at each AND
gate; its leaves are constant one or oracle variables, never constant zero.
Let $\Pi(r,w)$ be the finite set of these proof trees, and let $K(P)$ be the set
of distinct oracle keys occurring in a proof $P$.  Finiteness follows from the
acyclic interval recurrence, including the strict split used for star.  By
expanding the circuit into accepting proofs, we obtain
\[
 F_{r,w}=\bigvee_{P\in\Pi(r,w)}\ \bigwedge_{z\in K(P)} z.
\]
Thus the span circuit is a polynomial-size representation of a monotone DNF
with one term per symbolic accepting proof.  We use the circuit for symbolic
evaluation and the proof DNF for reasoning about witnesses and lower bounds.

\section{Sharp Two-Round Adaptivity}
\label{sec:characterization}

We first identify optimal semantic evaluation with decision-tree evaluation,
then construct an instance that attains the largest possible adaptivity gap
asymptotically and needs only two rounds.

Let $\mathcal K_{r,w}$ be the finite set of oracle keys appearing as input
variables of $F_{r,w}$.  For a Boolean function $f$, let $D(f)$ denote
deterministic decision-tree depth and let $D^{[R]}(f)$ denote the minimum
worst-case number of queries in an $R$-round deterministic decision tree.  Let
$\Rexp(f)$ denote the minimum,
over randomized decision trees with pointwise error at most $\delta$, of the
maximum over inputs of the expected number of queried variables.  This is the
expected-cost convention used throughout the paper; it differs from bounding
every randomized execution by a fixed depth~\cite{buhrman2002decisiontrees}.

\begin{theorem}[Decision-tree characterization]
\label{thm:decision-tree}
For every fixed expression $r$ and word $w$,
\[
 \Qdet(r,w)=D(F_{r,w})
 \quad\text{and}\quad
 \Qrand(r,w)=\Rexp(F_{r,w}).
\]
Moreover, the exact deterministic nonadaptive complexity equals the number of
essential variables of $F_{r,w}$.
\end{theorem}

\begin{proof}
We first eliminate queries outside $\mathcal K_{r,w}$.  Fix their answers to
zero.  Given any matcher $A$ and an assignment to $\mathcal K_{r,w}$, simulate
$A$ on the full oracle obtained by extending that assignment with these fixed
zeros.  Answer an out-of-support request locally rather than issuing an oracle
call.  Membership depends only on $\mathcal K_{r,w}$, so the simulation
preserves exact correctness or pointwise error and never increases the number
of actual queries.

An adaptive matcher restricted to $\mathcal K_{r,w}$ induces a decision tree
whose internal nodes are oracle variables, whose edges are answers, and whose
leaves return membership.  Uniform correctness requires this tree to compute
$F_{r,w}$, and conversely any decision tree for $F_{r,w}$ is an oracle matcher.
Fixing the private coins gives the randomized equivalence.  Finally, a
nonadaptive exact matcher must query every essential variable: omitting one
leaves two assignments that agree on all observed bits but have different
outputs.  Querying all essential variables is sufficient.
\end{proof}

Call a key \emph{essential} if changing only its value can change $F_{r,w}$,
and write $e(r,w)$ for the number of essential keys.  A basic counting bound
will provide the universal half of the sharp theorem.

\begin{lemma}[Support counting]
\label{lem:support-counting}
Every Boolean function with $E$ essential variables has deterministic
decision-tree depth at least $\lceil\log_2(E+1)\rceil$.
\end{lemma}

\begin{proof}
A depth-$d$ binary decision tree has at most $2^d-1$ internal nodes.  Every
essential variable must label an internal node of every tree computing the
function; otherwise the output would be independent of that variable.  Hence
$E\le2^d-1$.
\end{proof}

\subsection{A Calibrated Monotone Address Function}

Fix an integer $E\ge2$.  Let
\[
 k=k_E=
 \min\left\{
  1\le k\le E-1:
  E-k\le\binom{k}{\lfloor k/2\rfloor}
 \right\},
 \qquad
 p=\left\lfloor\frac{k}{2}\right\rfloor,
 \qquad
 m=E-k.
\]
The set in the minimum is nonempty because $k=E-1$ is feasible.  Choose an
arbitrary family
\[
 \mathcal A\subseteq\binom{[k]}p,
 \qquad |\mathcal A|=m.
\]
The function has address variables $x_1,\ldots,x_k$ and one data variable
$y_S$ for each $S\in\mathcal A$:
\begin{equation}
\label{eq:monotone-address}
 \operatorname{MA}_{k,\mathcal A}(x,y)
 =\operatorname{Thr}_{p+1}(x)
  \lor
  \bigvee_{S\in\mathcal A}
  \left(y_S\land\bigwedge_{i\in S}x_i\right),
\end{equation}
where $\operatorname{Thr}_{p+1}(x)=1$ exactly when $|x|\ge p+1$.
Equivalently,
\[
 \operatorname{MA}_{k,\mathcal A}(x,y)=
 \begin{cases}
  0, & |x|<p,\\
  y_S, & |x|=p\text{ and }\operatorname{supp}(x)=S\in\mathcal A,\\
  0, & |x|=p\text{ and }\operatorname{supp}(x)\notin\mathcal A,\\
  1, & |x|>p.
 \end{cases}
\]
This is a partial central-layer version of monotone addressing.  The full
function is classical and is explicitly discussed by Kane as a monotone
depth-$d$ function depending on exponentially many variables
\cite{kane2013monotone}.  Our use of it is the exact support calibration and
the restricted SemRE realization below.

\begin{lemma}[Exact two-round complexity]
\label{lem:address-complexity}
The function $\operatorname{MA}_{k,\mathcal A}$ has $E=k+m$ essential
variables and satisfies
\[
 D^{[1]}(\operatorname{MA}_{k,\mathcal A})=E,
 \qquad
 D^{[2]}(\operatorname{MA}_{k,\mathcal A})
 =D(\operatorname{MA}_{k,\mathcal A})=k+1.
\]
\end{lemma}

\begin{proof}
For the upper bound, query all $k$ address variables in the first round.  If
$|x|<p$, reject; if $|x|>p$, accept; and if $|x|=p$ with support outside
$\mathcal A$, reject.  The only remaining case has
$\operatorname{supp}(x)=S\in\mathcal A$, in which case querying $y_S$ in the
second round determines the answer.  Thus the worst-case cost is $k+1$.

For the lower bound, fix $S\in\mathcal A$.  The discrete derivative with
respect to $y_S$ is
\[
 \Delta_{y_S}\operatorname{MA}_{k,\mathcal A}
 =\prod_{i\in S}x_i\prod_{j\notin S}(1-x_j).
\]
Indeed, $y_S$ matters exactly at the address $x=1_S$.  The right-hand side has
a nonzero monomial of degree $k$, so the unique multilinear polynomial for
$\operatorname{MA}_{k,\mathcal A}$ has degree at least $k+1$.  Deterministic
decision-tree depth is at least exact polynomial degree
\cite{buhrman2002decisiontrees,jukna2012boolean}, proving
$D(\operatorname{MA}_{k,\mathcal A})\ge k+1$.

Every data variable $y_S$ is essential at $x=1_S$.  Every address variable
$x_i$ is essential after all data variables are set to zero: set exactly $p$
other address variables to one and flip $x_i$ across the threshold.  Hence all
$E$ variables are essential, and exact one-round evaluation costs $E$ by
Theorem~\ref{thm:decision-tree}.
\end{proof}

\subsection{A Linear-Size Unary SemRE}

We now realize Equation~\eqref{eq:monotone-address} with unit-length semantic
spans.  Use distinct predicate names and define
\[
 X_i=a\cap\langle q_i\rangle,
 \qquad
 Y_S=a\cap\langle q_S'\rangle.
\]
For $U\subseteq[k]$, let
\[
 P_U=Z_1^U\cdots Z_k^U,
 \qquad
 Z_i^U=
 \begin{cases}
  X_i,&i\in U,\\
  a,&i\notin U.
 \end{cases}
\]
Take $w_E=a^{k+1}$ and
\begin{equation}
\label{eq:address-semre}
 r_E=
 \bigcup_{T\in\binom{[k]}{p+1}} P_Ta
 \;\cup\!
 \bigcup_{S\in\mathcal A} P_SY_S.
\end{equation}
The first union computes the threshold term, while the second contains one
addressed data term for each $S\in\mathcal A$.

The flat expression has size $O(kE)$.  Factoring common prefixes removes the
factor $k$.  We use the following exact trie count.

\begin{lemma}[Constant-weight prefix trie]
\label{lem:constant-weight-trie}
The prefix trie of all binary strings containing exactly $u$ ones and $v$
zeros has
\[
 T(u,v)=\binom{u+v+2}{u+1}-1
\]
nodes, including the root.
\end{lemma}

\begin{proof}
The two first-symbol choices give
$T(u,v)=1+T(u-1,v)+T(u,v-1)$, with
$T(u,0)=u+1$ and $T(0,v)=v+1$.  Pascal's identity verifies the displayed
solution.
\end{proof}

Interpreting one as $X_i$ and zero as the literal $a$ at position $i$, a trie
becomes a factored regular expression with constant syntax overhead per node.
The threshold layer in Equation~\eqref{eq:address-semre} uses
$T(p+1,k-p-1)$ nodes.  The address prefixes of the data layer form a subtrie
of the full layer with $p$ ones and use at most $T(p,k-p)$ nodes, plus one
$Y_S$ leaf per $S\in\mathcal A$.  Both trie counts are
$O\bigl(\binom{k}{p}\bigr)$.

Minimality of $k$ gives
\[
 \binom{k}{p}\le 2\binom{k-1}{\lfloor(k-1)/2\rfloor}<2E
\]
for $k>1$, with the case $k=1$ immediate.  Thus the factored expression has
size $O(E)$.  Conversely, its $E$ distinct predicate names must all occur, so
$|r_E|=\Theta(E)$.

We can now state the headline theorem.

\begin{theorem}[Sharp two-round adaptivity]
\label{thm:sharp-adaptivity}
For every integer $E\ge2$, the unary instance $(r_E,w_E)$ above is star-free,
has semantic depth one, uses only unit-length semantic spans, and satisfies
\[
 e(r_E,w_E)=E,
 \qquad
 |r_E|=\Theta(E),
 \qquad
 \Qround{1}(r_E,w_E)=E,
\]
whereas
\[
 \Qround{2}(r_E,w_E)=\Qdet(r_E,w_E)=d_E
\]
with
\[
 d_E=\log_2E+\frac12\log_2\log_2E+O(1).
\]
Consequently,
\[
 \Gamma(E)=
 \max_{e(r,w)=E}
 \frac{\Qround{1}(r,w)}{\Qdet(r,w)}
 =(1+o(1))\frac{E}{\log_2E}.
\]
\end{theorem}

\begin{proof}
The construction induces $\operatorname{MA}_{k,\mathcal A}$, so all query and
essentiality claims follow from Lemma~\ref{lem:address-complexity}; the syntax
claims follow from the trie factorization above.

It remains to estimate $k$.  Let
$B_k=\binom{k}{\lfloor k/2\rfloor}$.  Feasibility and minimality give
$E-k\le B_k<2E+2$.  Since $B_k\ge2^k/(k+1)$, this first implies
$k=O(\log E)$ and hence $B_k=\Theta(E)$.  Stirling's formula gives
\[
 B_k=\Theta\!\left(\frac{2^k}{\sqrt{k}}\right).
\]
Taking logarithms and using $k=\Theta(\log E)$ yields
\[
 k=\log_2E+\frac12\log_2\log_2E+O(1).
\]
Lemma~\ref{lem:address-complexity} gives $d_E=k+1$, which changes only the
$O(1)$ term.

For an arbitrary instance with $E$ essential keys, exact one-round cost is
$E$, while Lemma~\ref{lem:support-counting} gives unrestricted cost at least
$\lceil\log_2(E+1)\rceil$.  Therefore
\[
 \Gamma(E)\le
 \frac{E}{\lceil\log_2(E+1)\rceil}
 =(1+o(1))\frac{E}{\log_2E}.
\]
The constructed instance gives the matching lower bound with the same leading
constant.
\end{proof}

\section{A Complete Round Hierarchy}
\label{sec:round-hierarchy}

The sharp construction shows that two rounds can already realize the largest
possible adaptivity gap.  We now repurpose the earlier guarded recursion to
answer a different question: how does the optimum vary at every intermediate
round budget?

Let
\[
 r_0=a\cap\langle q_z\rangle,
 \qquad
 f_0(z)=z,
 \qquad
 w_0=a.
\]
Given two renamed copies $r_t^{(0)}$ and $r_t^{(1)}$ with disjoint predicate
sets, introduce fresh one-letter filter atoms
\[
 X=a\cap\langle q_x\rangle,
 \qquad
 Y=a\cap\langle q_y\rangle
\]
and define, with $a^0=\eps$,
\begin{equation}
\label{eq:guarded-recursion}
 r_{t+1}=Yr_t^{(0)}\ \cup\ XYa^t\ \cup\ Xr_t^{(1)}.
\end{equation}
On $w_{t+1}=a^{t+2}$, the induced function is
\begin{equation}
\label{eq:guarded-function}
 f_{t+1}=y f_t^{(0)}\lor xy\lor x f_t^{(1)}.
\end{equation}
The two guard answers have a simple operational meaning:
\[
 (x,y)=(0,0)\Rightarrow0,
 \quad
 (0,1)\Rightarrow f_t^{(0)},
 \quad
 (1,0)\Rightarrow f_t^{(1)},
 \quad
 (1,1)\Rightarrow1.
\]

\begin{lemma}[Guarded-family baseline]
\label{lem:guarded-baseline}
For every $t\ge0$, the instance $(r_t,w_t)$ is unary, star-free, and of
semantic depth one; every semantic span has length one.  It has
\[
 e_t=3\cdot2^t-2
\]
essential keys, syntax size $\Theta(e_t)$, and exact unrestricted cost
\[
 \Qdet(r_t,w_t)=2t+1.
\]
\end{lemma}

\begin{proof}
The recurrence for the number of variables is $e_0=1$ and
$e_{t+1}=2e_t+2$, giving the displayed formula.  Every variable is essential:
the assignments $(x,y)=(0,1)$ and $(1,0)$ expose the left and right recursive
copies, while suitable accepting assignments in those copies expose the fresh
guards.

Querying $x$ and then $y$ either decides the output or selects one recursive
copy, so the recurrence $T_{t+1}\le T_t+2$ gives $T_t\le2t+1$.  For the
matching lower bound, let $p_t$ be the multilinear polynomial for $f_t$.  If
$A$ and $B$ are the polynomials for the two recursive copies, then
\[
 p_{t+1}=yA+xy+xB-xyA-xyB.
\]
The degree-$\deg(p_t)+2$ monomials in $xyA$ and $xyB$ use disjoint recursive
variable sets and cannot cancel.  Thus $\deg(p_t)=2t+1$, and decision-tree
depth is at least polynomial degree.  The syntax-size recurrence is
$S_{t+1}=2S_t+O(t)$ because of the padding $a^t$, so $S_t=O(2^t)$; the
$e_t$ distinct predicate names give the reverse bound.
\end{proof}

\begin{theorem}[Round hierarchy]
\label{thm:round-hierarchy}
For the guarded family and every $1\le R\le t$,
\[
 \Qround{R}(r_t,w_t)
 =\Theta\!\left(R\,2^{t/R}\right)
 =\Theta\!\left(R\,e_t^{1/R}\right).
\]
In particular,
\[
 \Qround{1}=\Theta(e_t),
 \qquad
 \Qround{2}=\Theta(\sqrt{e_t}),
 \qquad
 \Qround{3}=\Theta(e_t^{1/3}),
 \qquad
 \Qround{t}=\Theta(t).
\]
\end{theorem}

\begin{proof}
\emph{Upper bound.}
Set $b=\lceil t/R\rceil$.  In one round, query both guards at every recursive
node in the next $b$ levels of the current subinstance.  There are at most
$2^b-1$ such nodes, so the round uses $2(2^b-1)=O(2^b)$ queries.  The four
guard cases above either determine the output or identify exactly one
subinstance at the block boundary.  Repeat this for the first $R-1$ rounds.
The remaining height is at most $b$, so the final round queries all variables
of that subinstance, again using $O(2^b)$ queries.  The total is
$O(R2^b)=O(R2^{t/R})$.

\emph{Lower bound.}
Restrict the input at every recursive node to the promise
\[
 (x,y)\in\{(0,1),(1,0)\}.
\]
Equation~\eqref{eq:guarded-function} then becomes pointer evaluation on a
complete binary tree of height $t$: each internal node stores a direction bit,
and the output is the data bit at the reached leaf.  Any algorithm for $f_t$
must also solve this promised problem.

Consider one batch containing $q$ queries while the adversary has preserved a
completely untouched subtree of height $h$.  Let
$d=\lceil\log_2(q+1)\rceil$.  If $d<h$, the subtree has $2^d>q$ descendants at
depth $d$, so at least one descendant subtree receives no query from the
batch.  The adversary assigns the queried direction bits so that the pointer
path enters such a subtree; all queries outside it become irrelevant.  Thus
one batch can force progress through at most $d$ levels while leaving the
remaining subtree untouched.

Let $q_j$ be the batch size in round $j$ along the resulting adversarial
transcript.  If at some round
$\lceil\log_2(q_j+1)\rceil$ is at least the remaining height, the inequality
below is already satisfied.  Otherwise the adversary continues as above.  If
the sum were smaller than $t$ after all $R$ rounds, a nontrivial untouched
subtree would remain and two completions could choose different reached-leaf
values.  Correctness therefore requires
\[
 t\le\sum_{j=1}^R\left\lceil\log_2(q_j+1)\right\rceil.
\]
Writing $Q=\sum_j q_j$, concavity of the logarithm gives
\[
 t\le R+R\log_2\!\left(\frac QR+1\right),
\]
and hence
\[
 Q\ge R\left(2^{t/R-1}-1\right).
\]
When $t/R\ge2$, this is $\Omega(R2^{t/R})$.  When
$1\le t/R<2$, Lemma~\ref{lem:guarded-baseline} gives
$Q\ge2t+1=\Omega(R2^{t/R})$.  This proves the lower bound uniformly.

Finally, $e_t=\Theta(2^t)$ implies
$2^{t/R}=\Theta(e_t^{1/R})$ for the stated range of $R$.
\end{proof}

The hierarchy has a direct latency interpretation.  Let
$\operatorname{Rounds}_B(r,w)$ be the minimum worst-case number of rounds when
each round may contain at most $B$ independently specified atomic queries.

\begin{corollary}[Batch-cap tradeoff]
\label{cor:batch-cap}
For every $B\ge1$,
\[
 \operatorname{Rounds}_B(r_t,w_t)
 =\Theta\!\left(
 \max\left\{
  1,\frac{\log(e_t+1)}{\log(B+1)}
 \right\}
 \right).
\]
\end{corollary}

\begin{proof}
The pointer adversary with $q_j\le B$ gives
$t\le R\lceil\log_2(B+1)\rceil$, proving the lower bound together with the
trivial requirement of one round.

For the upper bound, if $B$ is bounded by a constant, the sequential strategy
from Lemma~\ref{lem:guarded-baseline} uses $O(t)$ rounds and matches the
formula.  Otherwise choose
$b=\Theta(\log(B+1))$ so that $3\cdot2^b\le B$.  Query a complete $b$-level
guard block in each round and query the final height-$b$ subinstance in the
last round.  This uses $O(\max\{1,t/b\})$ rounds, each respecting the cap.
\end{proof}

\section{A Fixed-Vocabulary Sparse-DNF Compiler}
\label{sec:fixed-vocabulary}

The two unary hard families use distinct predicate names because every
semantic span has value $a$.  We now encode Boolean variables by distinct
substring values instead.  A packed word writes each DNF literal only where
it is used, avoiding the quadratic syntax blowup of copying a complete list of
all variables for every term.

Let
\[
 f(z_1,\ldots,z_e)
 =\bigvee_{j=1}^M\ \bigwedge_{i\in S_j}z_i
\]
be a nonconstant monotone DNF after duplicate terms have been removed, and let
\[
 L=\sum_{j=1}^M |S_j|
\]
be its total number of literal occurrences.  We take the $z_i$ to be the
essential variables of $f$; constant functions have trivial encodings.

\begin{theorem}[Fixed-vocabulary sparse-DNF compiler]
\label{thm:sparse-dnf-compiler}
Over the fixed alphabet $\{0,1,\#\}$ and one predicate name $q$, there is a
SemRE instance $(r_f,w_f)$ whose membership function is exactly $f$ and such
that
\[
 |r_f|+|w_f|=O((M+L)\log(e+1)).
\]
The expression has semantic depth one, contains only two unfiltered Kleene
stars, and every filtered span has length
\[
 \ell=\lceil\log_2(e+1)\rceil.
\]
\end{theorem}

\begin{proof}
Choose distinct codes $c_1,\ldots,c_e\in\{0,1\}^{\ell}$, and let $C_i$ be the
literal regular expression spelling $c_i$.  List the indices in each term
$S_j=\{i_{j,1}<\cdots<i_{j,s_j}\}$.  Its framed word block and corresponding
expression block are
\begin{align*}
 W_j&=\#\#\,c_{i_{j,1}}\#c_{i_{j,2}}\#\cdots\#c_{i_{j,s_j}}\,\#\#,\\
 B_j&=\#\#\,
 (C_{i_{j,1}}\cap\langle q\rangle)\#
 (C_{i_{j,2}}\cap\langle q\rangle)\#\cdots\#
 (C_{i_{j,s_j}}\cap\langle q\rangle)\,\#\#.
\end{align*}
In both displays, a single $\#$ occurs only between consecutive codes; in
particular, a singleton term has the form
$\#\#c_i\#\#$ and $\#\#(C_i\cap\langle q\rangle)\#\#$.
Separate consecutive $W_j$ blocks by three additional copies of $\#$, and
let $w_f$ be their concatenation.  Finally, with
$A=0\cup1\cup\#$, define
\[
 r_f=A^*(B_1\cup\cdots\cup B_M)A^*.
\]

The framing prevents unintended cross-block matches.  Any occurrence of a
$B_j$ must start with two copies of $\#$ immediately followed by a binary
code.  In every separator run, this aligns the occurrence with the beginning
of one framed block.  The internal single-$\#$ separators and the final
double-$\#$ frame then force the complete ordered code sequence of that block.
Since the codes are distinct and term indices are sorted, $B_j$ matches a
framed block exactly when it matches the copy of term $S_j$ in the word.

The semantic filters in that match query precisely the keys
$\Ocal(q,c_i)$ for $i\in S_j$.  Extensionality identifies repeated copies of a
code across different blocks.  Therefore $B_j$ succeeds exactly when all
literals of term $j$ are one, and the union induces $f$.

Each literal occurrence contributes one length-$\ell$ code to the word and
one to the expression.  Frames, separators, and union nodes contribute
$O(M+L)$ additional syntax.  Hence
$|r_f|+|w_f|=O((M+L)\ell)$.  Only the two outer copies of $A^*$ use Kleene
star, and all filters are unnested and applied to one length-$\ell$ code.
\end{proof}

\subsection{The Sharp Family over One Predicate}

For the monotone-address function in
Equation~\eqref{eq:monotone-address}, the threshold contributes
$\binom{k}{p+1}$ terms and the data layer contributes $m$ terms.  Thus
\[
 M=\binom{k}{p+1}+m=O(E).
\]
Every term has exactly $p+1=O(\log E)$ literals, so
\[
 L=(p+1)M=O(E\log E).
\]

\begin{corollary}[One-predicate sharp gap]
\label{thm:fixed-vocabulary-gap}
For every $E\ge2$, there is a semantic-depth-one SemRE instance over
$\{0,1,\#\}$ using one predicate name such that
\[
 |r_E|+|w_E|=O(E\log^2E),
\]
every semantic span has length $\lceil\log_2(E+1)\rceil$, and
\[
 e(r_E,w_E)=E,
 \qquad
 \Qround{1}(r_E,w_E)=E.
\]
Moreover,
\[
 \Qround{2}(r_E,w_E)=\Qdet(r_E,w_E)
 =\log_2E+\tfrac12\log_2\log_2E+O(1).
\]
\end{corollary}

\begin{proof}
Apply Theorem~\ref{thm:sparse-dnf-compiler} to the DNF in
Equation~\eqref{eq:monotone-address}.  The compiler preserves the induced
Boolean function, so Theorem~\ref{thm:sharp-adaptivity} supplies the query
claims.
\end{proof}

The same compiler transfers the complete round hierarchy.  The guarded DNF
recurrences satisfy
\[
 M_{t+1}=2M_t+1,
 \qquad
 L_{t+1}=2L_t+2M_t+2,
\]
with $M_0=L_0=1$.  Hence
$M_t=O(e_t)$ and $L_t=O(e_t\log e_t)$.

\begin{corollary}[One-predicate round hierarchy]
\label{cor:fixed-vocabulary-hierarchy}
For every $t$, the guarded function $f_t$ has a semantic-depth-one realization
over $\{0,1,\#\}$ with one predicate name, logarithmic-length semantic spans,
and total representation size $O(e_t\log^2e_t)$.  For every
$1\le R\le t$, this realization satisfies
\[
 \Qround{R}=\Theta(R e_t^{1/R}).
\]
\end{corollary}

The logarithmic span length is unavoidable up to constants.  Over an alphabet
of size $\sigma$ with one predicate name, spans of length at most $s$ expose
at most
\[
 1+\sigma+\cdots+\sigma^s
\]
different oracle keys.  An instance with $E$ essential keys therefore needs
$s=\Omega(\log_\sigma E)$ for every fixed $\sigma\ge2$.  The compiler attains
this span length while keeping the representation within polylogarithmic
factors of linear size.

\section{Exact Randomization and Candidate Minimax}
\label{sec:lower-bounds}

We use pointwise error $\delta<1/2$ and worst-case expected query cost.  Under
this convention, nonadaptive randomization has an exact value for every
instance, not merely a lower bound.

\begin{theorem}[Exact randomized nonadaptive complexity]
\label{thm:random-nonadaptive}
For every SemRE instance $(r,w)$ and every $0\le\delta<1/2$,
\[
 \Qrandnonad(r,w)=(1-2\delta)e(r,w).
\]
\end{theorem}

\begin{proof}
Let $z_i$ be an essential variable.  Choose assignments $\alpha_i$ and
$\beta_i$ that differ only in $z_i$ and have different correct outputs.  Run a
randomized nonadaptive matcher on both assignments with the same private
coins.  Its random query set is independent of the answers.  Whenever that set
omits $z_i$, the two transcripts and outputs are identical, so at least one of
the two executions errs.  Therefore
\[
 \Pr[z_i\text{ is not queried}]
 \le
 \Pr[\operatorname{err}\mid\alpha_i]
 +\Pr[\operatorname{err}\mid\beta_i]
 \le2\delta.
\]
Summing the query indicators over all essential variables gives expected cost
at least $(1-2\delta)e(r,w)$.

For the matching upper bound, output zero without querying with probability
$\delta$, output one without querying with probability $\delta$, and with
probability $1-2\delta$ query all essential variables and evaluate exactly.
On every input, exactly one constant-output branch can be wrong, so the
pointwise error is $\delta$ and the expected cost is
$(1-2\delta)e(r,w)$.
\end{proof}

The same coupling gives a useful adaptive lower bound.  For an input $\alpha$
to a Boolean function $f$, let $s(f,\alpha)$ be the number of coordinates
whose individual flip changes $f(\alpha)$.

\begin{lemma}[Sensitivity under expected cost]
\label{lem:random-sensitivity}
Every pointwise-error-$\delta$ randomized decision tree satisfies
\[
 R_\delta^{\mathrm{exp}}(f)
 \ge (1-2\delta)s(f,\alpha)
\]
for every input $\alpha$.  Conversely,
\[
 R_\delta^{\mathrm{exp}}(f)\le(1-2\delta)D(f).
\]
\end{lemma}

\begin{proof}
For each sensitive coordinate $i$, couple the execution on $\alpha$ with the
execution on the input obtained by flipping $i$.  Until $i$ is queried, the
two adaptive transcripts coincide.  If the execution on $\alpha$ never
queries $i$, at least one coupled execution errs, so $i$ is queried on
$\alpha$ with probability at least $1-2\delta$.  Summing over sensitive
coordinates proves the lower bound.  The upper bound uses the same three-way
mixture as Theorem~\ref{thm:random-nonadaptive}, replacing exhaustive
evaluation by an optimal exact decision tree.
\end{proof}

At a middle-layer address $x=1_S$, set all nonaddressed data bits to zero.
Setting $y_S=0$ gives sensitivity $k-p+1$, while setting $y_S=1$ gives
sensitivity $p+1$.  Lemma
\ref{lem:random-sensitivity} and the exact depth $k+1$ therefore imply the
following robustness statement.

\begin{corollary}[Randomized adaptivity gap]
\label{cor:random-address-gap}
For the sharp family and every fixed $\delta<1/2$,
\[
 \Qrand(r_E,w_E)=\Theta(\log E),
 \qquad
 \Qrandnonad(r_E,w_E)=(1-2\delta)E.
\]
Thus the $\Theta(E/\log E)$ adaptive advantage survives bounded error.
\end{corollary}

The separation is not confined to exact worst-case essentiality.  In the full
central-layer family $\mathcal A=\binom{[k]}p$, choose $S$ uniformly, set
$x=1_S$, set all nonaddressed data bits to zero, and choose $y_S$ uniformly.
Then the output is $y_S$.  If $Q_y$ is the random set of queried data
variables, it hits $y_S$ with probability
$\mathbb E[|Q_y|]/|\mathcal A|$.  Conditional on a miss, every output has error
at least $1/2$.  Hence average error at most $\delta$ requires at least
$(1-2\delta)|\mathcal A|=\Omega(E)$ expected data queries, while the exact
two-round algorithm still uses $k+1=O(\log E)$ queries.

\subsection{Exact Fixed-Word Minimax Values}

For a word $w$ of length $n$, let
\[
 \Sub(w)=\{w[i:j]:0\le i\le j\le n\},
 \qquad
 \SD(w)=|\Sub(w)|.
\]
We call $\SD(w)$ the \emph{substring diversity} of $w$.  Extensional caching
makes substring values, rather than their positions, the relevant unit.  For a
span bound $s$, define
\[
 \SD_s(w)=
 \bigl|\{w[i:j]:0\le i\le j\le n,\ j-i\le s\}\bigr|.
\]

Let $A=\bigcup_{a\in\Sigma}a$.  The factor-OR expression
\[
 r_h=\bigcup_{j=1}^h
 A^*\bigl(A^*\cap\langle q_j\rangle\bigr)A^*
\]
induces the OR of all $h\SD(w)$ candidate keys, including the empty substring.

\begin{theorem}[Exact candidate-universe minimax]
\label{thm:candidate-universe}
Fix a word $w$, an integer $h\ge1$, and $0\le\delta<1/2$.  Then
\[
 \max_{r:\,|\operatorname{Pred}(r)|\le h}\Qdet(r,w)
 =h\SD(w)
\]
and
\[
 \max_{r:\,|\operatorname{Pred}(r)|\le h}\Qrand(r,w)
 =(1-2\delta)h\SD(w).
\]
Both maxima are attained by the semantic-depth-one expression $r_h$, whose
syntax size is $O(h|\Sigma|)$.
\end{theorem}

\begin{proof}
An expression using at most $h$ predicate names can depend only on keys
$(q,x)$ with $x\in\Sub(w)$, so there are at most $N=h\SD(w)$ candidates.
Querying all candidates and evaluating the span circuit gives deterministic
cost at most $N$.  The three-way mixture that outputs zero, outputs one, or
performs this exact evaluation gives randomized expected cost at most
$(1-2\delta)N$.

For $r_h$, the filtered middle factor can be any substring of $w$, including
$\eps$.  Hence
\[
 F_{r_h,w}=\bigvee_{j=1}^h\ \bigvee_{x\in\Sub(w)}X_{q_j,x},
\]
the OR of exactly $N$ independent variables.  On the all-zero oracle, exact
evaluation must query every variable.  Moreover, all $N$ coordinates are
sensitive there, so Lemma~\ref{lem:random-sensitivity} gives randomized
expected cost at least $(1-2\delta)N$.  The displayed syntax repeats the
constant-size factor skeleton for each predicate, giving size
$O(h|\Sigma|)$.
\end{proof}

The same argument is exact under a semantic-span bound.

\begin{corollary}[Exact bounded-span minimax]
\label{thm:span-lb}
Fix $w$, $h\ge1$, $s\ge0$, and $\delta<1/2$.  Among SemREs using at most $h$
predicate names and whose feasible semantic spans have length at most $s$,
\[
 \max_r\Qdet(r,w)=h\SD_s(w)
\]
and
\[
 \max_r\Qrand(r,w)=(1-2\delta)h\SD_s(w).
\]
\end{corollary}

\begin{proof}
The candidate universe contains at most $h\SD_s(w)$ keys.  For the matching
lower bounds, let
\[
 U_s=\bigcup_{\ell=0}^s A^\ell,
 \qquad
 r_{h,s}=\bigcup_{j=1}^h
 A^*(U_s\cap\langle q_j\rangle)A^*.
\]
This semantic-depth-one expression induces the OR of exactly the
$h\SD_s(w)$ keys whose substring length is at most $s$.  Apply the deterministic
all-zero adversary and the sensitivity lower bound.
\end{proof}

For the word $0^m1^m$, every distinct substring is one of $\eps$, $0^i$,
$1^j$, or $0^i1^j$, and therefore $\SD(0^m1^m)=(m+1)^2$.  This recovers the
factor-OR mechanism of Huang et al.~\cite{huang2025membership} and upgrades it
to an exact value for every fixed word, with duplicate substring values
removed.  On substring-rich binary words the value is $\Theta(hn^2)$; on the
unary word $a^n$, extensional caching reduces it to $h(n+1)$.

\section{Related Work}
\label{sec:related}

\subsection{Semantic Regular Expressions}

Chen et al.\ introduced semantic regular expressions and the \textsc{Smore}
system, with the goal of synthesizing extraction patterns from positive and
negative examples~\cite{chen2023semantic}.  Their work established SemREs as a
programming abstraction that combines regular structure with semantic
predicates, but did not study the oracle-query complexity of membership.
Huang et al.\ gave the first detailed analysis of SemRE membership
\cite{huang2025membership}.  They presented a two-pass NFA/query-graph
algorithm, analyzed its symbolic running time, related nested matching to
triangle finding, and proved an $O(|r||w|^2)$ query upper bound for their
algorithm.  They also established worst-case lower bounds of
$\Omega(|w|^2)$ for a fixed nonempty predicate vocabulary and
$\Omega(|r||w|^2)$ when the vocabulary is unbounded.  Their hard family uses a
factor-OR expression on $0^m1^m$ and forces the keys $(q_i,0^j1^k)$.

The present paper asks a complementary question.  Rather than analyze the
queries made by one matcher, we minimize over all correct matchers after both
$r$ and $w$ have been fixed, and we refine adaptivity by its number of parallel
rounds.  This viewpoint exposes essential keys, exact randomized expected
cost, and the complete latency--work tradeoff.  Our candidate-universe theorem
identifies the exact deterministic and randomized per-word minimax values
$h\SD(w)$ and $(1-2\delta)h\SD(w)$.  On $0^m1^m$, where
$\SD(0^m1^m)=(m+1)^2$, this recovers the multi-predicate factor-OR family of
Huang et al.; the new point is the exact value for arbitrary words, including
deduplication of repeated substring values and the bounded-span refinement
$h\SD_s(w)$.  Conversely, their query graph could replace our cubic span chart
as a local representation without changing the Boolean function induced by a
fixed instance.

Corollary~\ref{thm:fixed-vocabulary-gap} further separates the adaptivity gap
from vocabulary growth: one predicate over a fixed alphabet suffices when
distinct logarithmic-length span values encode the oracle keys.  The packed
sparse-DNF compiler reduces the total representation size to
$O(E\log^2E)$.  The linear-size unary construction and the one-predicate
construction therefore expose two different points in the tradeoff between
vocabulary, span length, and syntax size.

Several other systems combine automata-like structure with language-model
judgments.  \textsc{ReLM} compiles validation queries over language-model
outputs into automata~\cite{kuchnik2023relm}, while $L^*LM$ places
natural-language oracles inside active automata learning
\cite{vazquezchanlatte2024lstarlm}.  Prompting languages similarly expose
control and data flow around model calls~\cite{beurerkellner2023prompting}.
These systems enable stable, cacheable semantic interfaces.  Our focus is
narrower: the number of independent predicate--substring judgments needed for
one SemRE membership instance.

\subsection{Classical and Extended Regular Expressions}

Classical regular-expression membership can be implemented through finite
automata.  Thompson's construction underlies linear-time simulation in the
input length for a fixed expression~\cite{thompson1968regex}.  Brzozowski
derivatives and Antimirov partial derivatives give alternative algebraic
routes from expressions to matchers
\cite{brzozowski1964derivatives,antimirov1996partial}, and systems such as
\textsc{Hyperscan} optimize automata for modern hardware
\cite{wang2019hyperscan}.  Fine-grained analyses identify sharp complexity
boundaries among classical pattern types
\cite{backurs2016regex,bringmann2017dichotomy}, while streaming formulations
expose additional time--space tradeoffs~\cite{dudek2022streaming}.

These results keep the language semantics fixed and charge local computation.
A SemRE predicate, by contrast, may depend on a whole span whose starting point
was selected earlier in the parse.  Consequently, a matcher may have to
distinguish quadratically many endpoint pairs.  More fundamentally,
$A^*\cap\langle q\rangle$ recognizes the language implemented by $q$, so an
arbitrary semantic predicate cannot in general be compiled into an
oracle-independent finite automaton.

Intersection and complement preserve regular expressiveness but can make
extended regular expressions much more succinct and their membership problem
substantially harder.  Prior work gives complexity classifications and
specialized algorithms for such expressions
\cite{petersen2002intersection,kupferman2002extended,rosu2007extended}; recent
derivative-based engines show that several extended operators can nevertheless
be supported efficiently in practice~\cite{varatalu2025resharp}.  Semantic
filtering resembles intersection syntactically, but its right operand is an
unknown span predicate rather than a regular language supplied as part of the
expression.

\subsection{Provenance and Boolean-Function Evaluation}

The proof-DNF expansion of the span circuit is closely related to database
lineage and provenance: both record alternative derivations and the input facts
used by each derivation.  Provenance semirings provide an algebraic framework
for such annotations~\cite{green2007provenance}.  Our circuit serves a narrower
purpose.  Its variables are extensional predicate--substring judgments, and we
ask how many of their values must be acquired adaptively rather than how to
propagate a general annotation.

Stochastic Boolean-function evaluation studies adaptive testing when variables
have costs and probability distributions, often seeking approximation
algorithms for expected cost~\cite{deshpande2016stochastic}.  Our objective is
different: correctness is pointwise over oracle assignments and cost is
worst-case deterministic or worst-case expected over those assignments.  The
weighted recurrence in Appendix~\ref{sec:algorithms} accommodates nonuniform
query costs but does not assume an input distribution.

\subsection{Decision-Tree and Parameterized Complexity}

Decision-tree complexity studies how many input variables must be inspected to
evaluate a Boolean function under deterministic, randomized, and quantum query
models~\cite{buhrman2002decisiontrees,jukna2012boolean}.  Classical work
relates decision trees to parallel computation~\cite{nisan1991crew} and
analyzes randomized formula and game-tree evaluation
\cite{saks1986probabilistic}.  Bounded rounds of adaptivity have also been
studied explicitly in property testing and randomized decision trees
\cite{canonne2018adaptivity}.  Those results establish hierarchies in different
query problems; our theorem gives the exact tradeoff
$\Theta(RE^{1/R})$ for one restricted SemRE family and converts it into an RPC
batch-cap bound.

The Boolean core of our sharp construction is monotone addressing.  Kane
records the classical function as a monotone depth-$d$ decision tree depending
on exponentially many variables~\cite{kane2013monotone}.  We do not claim this
Boolean idea as new.  Our contributions are the calibration to every support
size $E$, the sharp extremal ratio, and the linear-size unary SemRE realization.
The counting, polynomial-degree, sensitivity, and pointer-subtree arguments are
standard query-complexity tools; the SemRE-specific work is to make their
premises hold under star-free, depth-one, unit-span restrictions and under a
one-predicate fixed-vocabulary encoding.

Parameterized complexity separates arbitrary dependence on a parameter from
polynomial dependence on the input size
\cite{downey2013parameterized,cygan2015parameterized}.  We borrow this
separation-of-dependence viewpoint only to distinguish query bounds that are
linear in the word length from bounds independent of it.  Since local
computation remains polynomial relative to the oracle, we do not introduce a
new machine-time complexity class.

\section{Concluding Remarks}
\label{sec:conclusion}

The main conclusion is sharper than an order-of-growth separation.  For every
number $E$ of essential oracle keys, a unary, star-free, semantic-depth-one
SemRE of size $\Theta(E)$ has one-round cost $E$ but exact two-round cost
\[
 \log_2E+\tfrac12\log_2\log_2E+O(1).
\]
Together with the universal support-counting bound, this determines the
largest nonadaptive-to-adaptive ratio as
$(1+o(1))E/\log_2E$.  Maximal savings therefore do not require a long adaptive
computation: two rounds suffice.

The guarded family gives the complementary picture.  Its exact $R$-round cost
is $\Theta(RE^{1/R})$, producing a strict hierarchy from exhaustive
nonadaptive evaluation to logarithmic unrestricted cost.  Under a cap of $B$
queries per interaction, this becomes an exact asymptotic latency law of
$\Theta(\max\{1,\log E/\log(B+1)\})$ rounds.  These two constructions separate
the existence of a large adaptivity gap from the shape of the full
round--query tradeoff.

The representation results show that neither phenomenon depends on a growing
vocabulary.  A packed sparse-DNF compiler realizes both families over the
fixed alphabet $\{0,1,\#\}$ with one predicate, logarithmic-length semantic
spans, and $O(E\log^2E)$ total size.  The span length is optimal up to constants
for a fixed alphabet.  Randomization is exact under the paper's expected-cost
convention: nonadaptive complexity is $(1-2\delta)E$ for every instance, and
the fixed-word candidate-universe minimax value is
$(1-2\delta)h\SD(w)$, or $(1-2\delta)h\SD_s(w)$ under a span bound.

The span circuit is the enabling SemRE representation.  It preserves sharing
between equal substring values and turns semantic matching into evaluation of
a monotone Boolean function without expanding all parses.  It is not a claim
of optimal local parsing time: symbolic computation, atomic oracle work, and
sequential oracle latency remain distinct resources.

\smallskip
\noindent\textbf{Scope and Directions.}
The black-box model deliberately assumes no relationship between distinct
oracle keys.  Correlations, logical implications, indexed lookup, and shared
predicate structure can reduce the lower bounds only when exposed through a
promise or a richer interface.  A natural next step is to characterize which
structured-oracle promises preserve the sharp round hierarchy.  Semantic
negation would replace monotone span circuits by general Boolean circuits,
while richer regular-expression operators would mix semantic-query savings
with their own succinctness costs.

\appendix
\section{Additional Consequences for Natural Parameters}
\label{sec:parameters}

The candidate-universe theorem treats every predicate--substring pair as
potentially relevant.  A particular instance may expose far fewer pairs
because its regular structure rules out most filter spans or because repeated
span values share an extensional key.  We now relate these progressively
smaller candidate sets to the optimal cost.

Let
$t$ be the number of filter occurrences and $h\le t$ the number of distinct
predicate identifiers.  Let $c_{\rm occ}$ count the
filter-occurrence/interval triples that survive sound symbolic pruning, and
let $c$ be the number of distinct keys $(q,w[i:j])$ among those triples.  Write
$e=e(r,w)$ for the essential-key count from Section~\ref{sec:characterization}.
Then
\[
 \Qdet(r,w)\le e\le c\le c_{\rm occ},
 \qquad c\le h\SD(w).
\]
Here $e$ is the exact Boolean support size, whereas $c$ is an operational
overapproximation obtainable from symbolic provenance.  Querying the $c$ keys
and then evaluating the span circuit proves the upper bound.  The factor-OR
construction has $e=c$, so every inequality can be tight.

\subsection{Symbolic Pruning and Candidate Feasibility}

Replace every semantic filter $u\cap\langle q\rangle$ by $u$, or equivalently
set every oracle variable to one.  If the root of this all-one regular
relaxation is false, the instance can be rejected without an oracle call.
Dually, if the circuit is true when every oracle variable is zero, membership
is oracle-independent and the matcher can accept immediately.  In every
remaining case, an inside/outside provenance traversal may discard filter
spans that occur on no all-one accepting proof.  This pruning is sound because
every actual accepting parse is also an all-one parse.

Provenance nevertheless overapproximates Boolean relevance.  For example,
$X_q\lor(X_q\land X_{q'})$ simplifies to $X_q$, although $q'$ appears on an
all-one accepting parse.  Computing exact relevance may itself require
nontrivial Boolean reasoning; our $c$ bound remains valid for any sound
candidate overapproximation.

\subsection{Bounded Semantic Spans}

Let $s$ be the maximum length of a span that remains symbolically feasible at
any filter.  This is a structural property of the instance, rather than the
largest span inspected by one particular query order.  Define
\[
 \SD_s(w)=\bigl|\{w[i:j]:j-i\le s\}\bigr|.
\]
For $\sigma=|\Sigma|$,
\[
 \SD_s(w)\le 1+\sum_{\ell=1}^{\min(s,n)}
          \min\{n-\ell+1,\sigma^\ell\}.
\]
Consequently, eager evaluation satisfies
\[
 Q\le h\SD_s(w)
 \quad\text{and}\quad
 Q=O(tn(s+1))
\]
where the second bound does not use value-based sharing across positions.  For
fixed $\sigma\ge2$, $\SD_s(w)=O(\sigma^s)$; for a unary alphabet,
$\SD_s(w)\le\min(n,s)+1$.  A fixed alphabet and span bound therefore cap the
value universe independently of input length, once every short string occurs.

Corollary~\ref{thm:span-lb} shows that these candidate bounds are exact in both
the deterministic and randomized minimax models.  Its witness uses
\[
 U_s=\bigcup_{\ell=0}^{s}A^\ell
 \quad\text{and}\quad
 r_{h,s}=\bigcup_{j=1}^{h}A^*(U_s\cap\langle q_j\rangle)A^*,
\]
where $A^0=\eps$.  This depth-one expression induces the OR of exactly the
$h\SD_s(w)$ keys $(q_j,x)$ with $|x|\le s$.  In the primitive grammar, one copy of
$U_s=\bigcup_{\ell=0}^{s}A^\ell$ has size
$O(|\Sigma|\sum_{\ell=0}^{s}(\ell+1))=O(|\Sigma|(s+1)^2)$.
Thus $r_{h,s}$ has size $O(h|\Sigma|(s+1)^2)$.

For $s\ge1$, a linearized $\sigma$-ary de Bruijn word of order $s$ contains all
$\sigma^s$ strings of length $s$.  Corollary~\ref{thm:span-lb} therefore yields
the tight lower bound $\Omega(h\sigma^s)$, so exponential dependence on $s$ is
unavoidable in the value universe.  When $w$ is substring-rich and $s$ grows
with $n$, the positional $\Theta(hns)$ regime may instead determine the count.
Written in the primitive grammar, $r_{h,s}$ has size polynomial in $h$, $s$,
and $|\Sigma|$ by the explicit estimate above; the query lower bound does not
hide an exponentially large expression.

\subsection{Nesting Depth and Positive Assignments}

Bounding semantic span length reduces the candidate universe; bounding nesting
depth alone does not.  In fact, distinct nested predicates can multiply the
number of independent keys.  Let
$\Sub^+(w)=\Sub(w)\setminus\{\eps\}$ and
$\SD^+(w)=|\Sub^+(w)|$.  Let
$B_0=A^+$, $B_j=B_{j-1}\cap\langle q_j\rangle$, and
$r_d=A^* B_d A^*$.  Then
\[
 F_{r_d,w}=\bigvee_{x\in\Sub^+(w)}
             \bigwedge_{j=1}^d X_{q_j,x}.
\]

\begin{theorem}[Nesting lower bound]
\label{thm:nesting}
For distinct $q_1,\ldots,q_d$,
$D(F_{r_d,w})=d\SD^+(w)$.
\end{theorem}

\begin{proof}
Querying all variables gives the upper bound.  For the lower bound, treat the
$d$ variables for each substring as one block.  An adversary answers one to
the first $d-1$ queries in a block and zero to its final query.  An incomplete
block has an accepting completion in which all remaining bits are one.  As
long as some block is incomplete, there is also a rejecting completion by
placing a zero in every incomplete block.  Hence a correct algorithm must
complete every block and query all $d\SD^+(w)$ bits.
\end{proof}

Repeated nesting with the same predicate adds no calls because the key
$(q,x)$ is cached.  On the other hand, depth one already permits quadratic
hardness, so a depth bound does not remove the base lower bound.

The number of positive keys requires a different quantifier order from the
structural parameters above.  For a matcher $A$ and oracle assignment $\alpha$, let
$Q_A(r,w;\alpha)$ be the number of queries made on that assignment and let
$p_\alpha$ be the number of positive keys among the $c$ candidates.  A bound in
terms of $p_\alpha$ is therefore a pointwise or promise-sensitive guarantee,
not a bound on the assignment-independent quantity $\Qdet(r,w)$.

Even under the promise $p_\alpha\le1$, exact factor-OR evaluation has
deterministic minimax cost $c$.  The promise includes the all-zero assignment,
and an adversary may place the unique positive at the last queried key.
The randomized coupling in Theorem~\ref{thm:candidate-universe} applies to the
same promise class and gives at least $(1-2\delta)c$ expected queries on the
all-zero assignment.  Thus an $O(p_\alpha)$ guarantee requires a stronger
interface that enumerates positives or supports indexed lookup.

\subsection{Parameter-Sensitive Query Bounds}

Classical parameterized complexity classifies ordinary running time
\cite{downey2013parameterized}.  With unit-cost oracle access, the span
algorithm is polynomial relative to the oracle regardless of the parameters
above.  Standard running-time classes therefore do not express when semantic
calls become linear in the input length or independent of it.  Borrowing only
the separation-of-dependence viewpoint, we use the following two descriptions.
A parameter $\kappa$ gives a \emph{linear-in-$n$ query bound} if, for some
computable function $f$ and fixed polynomial $P$,
\[
 Q\le f(\kappa)P(m,\sigma)n.
\]
It gives a \emph{string-independent query bound} if
\[
 Q\le f(\kappa)P(m,\sigma),
\]
where the degree of $P$ is independent of $\kappa$.  These phrases describe
bound shapes; we do not introduce new parameterized complexity classes.

\begin{corollary}[Parameter-sensitive query-bound summary]
\label{thm:parameter-classification}
Under the atomic oracle model:
\begin{enumerate}
  \item The essential-key count gives $Q\le e$, and the operational candidate
  count gives $Q\le c$, with $e\le c$.  Both bounds are worst-case tight.
  \item The span bound $s$ gives the linear-in-$n$ bound
  $Q=O(tn(s+1))$.
  \item For a fixed alphabet, $s$ also gives a string-independent bound.
  Dependence $\Theta(\sigma^s)$ is unavoidable when $\sigma\ge2$, while
  $\Theta(s+1)$ is tight for a unary alphabet.
  \item None of $h$, $t$, $|r|$, or nesting depth $d$ alone guarantees a
  linear-in-$n$ query bound.
\end{enumerate}
If $\sigma$ is part of the input, $s$ alone does not give a
string-independent bound of the form
$f(s)\sigma^C\operatorname{poly}(m)$ for a constant $C$ independent of $s$.
\end{corollary}

\begin{proof}
Item 1 follows from Theorem~\ref{thm:decision-tree} and the candidate upper
bound; items 2 and 3 follow from Corollary~\ref{thm:span-lb}.  The unary
dependence in item 3 is tight because
$\SD_s(a^n)=s+1$ whenever $n\ge s$.  For item 4, fix the listed parameter to a
constant: the depth-one, one-predicate expression from
Theorem~\ref{thm:candidate-universe} still costs $\Theta(n^2)$ on
substring-rich words.  The
final statement follows from the $\Omega(\sigma^s)$ de Bruijn instance,
which cannot be bounded by $f(s)\sigma^C$ for a universal constant $C$.
\end{proof}

\begin{table}[t]
\caption{Parameter-sensitive query bounds for exact matching.}
\label{tab:parameters}
\small
\begin{tabular}{@{}p{0.17\columnwidth}p{0.75\columnwidth}@{}}
\toprule
Parameter & Query-sensitive conclusion \\
\midrule
$e$ & Exact nonadaptive cost; deterministic adaptive cost lies between
$\lceil\log_2(e+1)\rceil$ and $e$.  The sharp family attains ratio
$(1+o(1))e/\log_2e$, with the optimal leading constant. \\
$c_{\rm occ}$ & $Q\le c_{\rm occ}$; duplicate values may make $c$ much
smaller. \\
$c$ & Independent of $n$ with $Q\le c$; the bound is worst-case tight. \\
$s$ & Linear in $n$ via $O(tn(s+1))$; independent of $n$ for fixed alphabets.
The unavoidable dependence is $\Theta(\sigma^s)$ for fixed $\sigma\ge2$ and
$\Theta(s+1)$ for $\sigma=1$. \\
$h$ or $t$ & No linear-in-$n$ bound by itself: one predicate and one filter
permit $\Theta(n^2)$. \\
$|r|$ & No linear-in-$n$ bound by itself: a constant expression permits
$\Theta(n^2)$. \\
$p_\alpha$ & Assignment-sensitive rather than structural: even the promise
$p_\alpha\le1$ leaves deterministic minimax cost $c$ and randomized
$\Omega(c)$ expected cost on the all-zero assignment. \\
$d$ & No linear-in-$n$ bound by itself; $d=1$ is quadratic, and distinct
nesting can cost $\Theta(dn^2)$. \\
\bottomrule
\end{tabular}
\end{table}

\section{Additional Algorithmic Consequences}
\label{sec:algorithms}

The circuit characterization defines the semantic optimum independently of
the local data structure used to realize it.  We record eager pruning, lazy
evaluation and its certificate structure, exact strategy synthesis, and
extensions beyond unit-cost atomic queries.  These procedures trade symbolic
planning for possible query savings; they do not improve the symbolic
complexity of SemRE matching.

\subsection{Eager Symbolic Pruning}

An eager matcher first constructs the span circuit and evaluates its all-zero
and all-one completions.  If these completions agree, the answer is independent
of the oracle.  Otherwise, the matcher performs contextual provenance pruning,
deduplicates the surviving candidates by their full keys $(q,x)$, queries the
remaining $c$ keys, and evaluates the exact circuit.  Queries may be grouped
into batches without changing their atomic count.  The resulting bounds are
\[
 Q\le c,
 \qquad
 T_{\rm sym}=O(mn^3),
 \qquad
 S=O(mn^2).
\]
Every actual accepting parse survives the all-one relaxation, which proves the
soundness of the pruning step.  The algorithm is query-optimal on negative
factor-OR instances, although it may retain variables that stronger Boolean
simplification would show to be irrelevant.

\subsection{Lazy Lower/Upper Evaluation}

Let $\rho$ be a partial assignment of oracle variables, and define
\[
 \begin{aligned}
 L_\rho&=F_{r,w}(\rho,\text{unknown}=0),\\
 U_\rho&=F_{r,w}(\rho,\text{unknown}=1).
 \end{aligned}
\]
Monotonicity gives
$L_\rho\le F_{r,w}(\Ocal)\le U_\rho$.  A lazy matcher accepts as soon as
$L_\rho=1$ and rejects as soon as $U_\rho=0$.  If neither condition holds, it
chooses an accepting proof under the all-one completion and queries an unknown
key on that proof.  Each iteration fixes a new variable, so the procedure
terminates after at most $c$ queries.

The proof-DNF view gives an exact description of the certificates that can
terminate this procedure.  Let $\mathcal H_{r,w}$ be the hypergraph whose
vertices are oracle keys and whose hyperedges are the sets $K(P)$ for symbolic
accepting proofs $P\in\Pi(r,w)$.

\begin{theorem}[Certificate hypergraph]
\label{thm:certificates}
Fix an oracle assignment $\alpha$.
\begin{enumerate}
  \item If $F_{r,w}(\alpha)=1$, its minimum 1-certificate size is
  \[
    \min\{|K(P)|:P\in\Pi(r,w),\ K(P)\subseteq\alpha^{-1}(1)\}.
  \]
  \item If $F_{r,w}(\alpha)=0$, its minimum 0-certificate size is the minimum
  number of zero-valued vertices that hit every hyperedge of
  $\mathcal H_{r,w}$.
\end{enumerate}
\end{theorem}

\begin{proof}
A set of positive variables forces a monotone DNF to one exactly when it
contains every variable of some term: if it contains no term, setting all
other variables to zero is a rejecting completion.  Dually, a set of negative
variables forces the DNF to zero exactly when it intersects every term;
otherwise, an unhit term can be completed to all ones.  Substituting the proof
terms $K(P)$ proves both statements.
\end{proof}

A short positive certificate can lead to early acceptance.  The converse is
not automatic: repeatedly following failed proofs may use many more queries
than the smallest certificate.  Recomputing the entire circuit after every
query gives the conservative symbolic bound $O(cmn^3)$; incremental
provenance maintenance can reduce this overhead without changing the query
guarantee.

For factor-OR, the minimum 1-certificate has size one, whereas every
0-certificate contains the full candidate set.  Thus a favorable positive
instance may stop after one successful query while an exact negative instance
remains exhaustive.

Nested filters admit another exact evaluation order.  Before querying the
outer predicate of $u\cap\langle q\rangle$, a matcher may compute the child
relation and retain only feasible spans.  This bottom-up order avoids outer
queries on semantically impossible children.  The opposite order can be better
when the outer predicate is cheap or likely to reject.  Neither order dominates
for all oracles, and Theorem~\ref{thm:decision-tree} captures the optimal
deterministic adaptive choice.

\subsection{Exact Strategy Synthesis}

The decision-tree characterization is constructive, although generic optimal
planning is exponential in the number of circuit variables.  Let $N$ be the
number of gates and $c$ the number of input variables.

\begin{proposition}[Optimal-strategy synthesis]
\label{prop:strategy-synthesis}
Given an $N$-gate monotone circuit with $c$ variables, an optimal deterministic
query strategy can be synthesized in $O(3^c(N+c))$ time and
$O(3^c+N)$ space.
\end{proposition}

\begin{proof}
For every partial assignment $\rho\in\{0,1,*\}^c$, let $G_\rho$ be the
restricted function.  By monotonicity, $G_\rho$ is constant exactly when its
all-zero and all-one completions agree; two circuit evaluations test this
condition in $O(N)$ time.  For each state, define
\[
 T(\rho)=
 \begin{cases}
 0,&G_\rho\text{ is constant},\\
 1+\displaystyle\min_{x:\rho(x)=*}\max_{b\in\{0,1\}}
 T(\rho[x\leftarrow b]),
 &\text{otherwise}.
 \end{cases}
\]
Evaluate the states in decreasing order of the number of assigned variables,
so that both children of every recurrence have already been computed.  There
are $3^c$ states.  Constant testing costs $O(N)$ per state, and minimizing over
the unassigned variables costs $O(c)$ per state.  Storing one value and one
minimizing variable per state gives the stated space bound.  The recurrence is
the standard optimal decision-tree recurrence, and following the stored
variables yields an optimal matcher.
\end{proof}

Query optimality need not minimize end-to-end cost.  If querying variable $x$
has weight $\lambda_x$, replacing the leading one in the recurrence by
$\lambda_x$ minimizes worst-case additive semantic cost instead of call count.

\subsection{Beyond Unit-Cost Atomic Queries}

When implementation costs matter, we report
\[
 (T_{\rm sym},Q_{\rm atom},R_B),
\]
where $T_{\rm sym}$ is local symbolic work, $Q_{\rm atom}$ is the number of
distinct semantic keys evaluated, and $R_B$ is the number of sequential RPC
interactions when one interaction may submit at most $B$ independently
specified keys.  These quantities separate local parsing, semantic work, and
latency.  A nonnegative weight $\lambda(q,x)$ gives the more general additive
semantic cost
\[
 C_{\Ocal}=\sum_{(q,x)\text{ queried}}\lambda(q,x).
\]
The unit-cost results take $\lambda\equiv1$.

Suppose one batch returns at most $B$ independently specified bits.  Eager
evaluation then uses $\lceil c/B\rceil$ rounds and $c$ atomic evaluations;
negative factor-OR instances make both bounds tight.  For instances with useful
partial adaptivity, this eager latency can be far from optimal:
Corollary~\ref{cor:batch-cap} gives the guarded family an exact
$\Theta(\max\{1,\log e/\log(B+1)\})$ round bound.  If an API instead returns
all $h$ predicates for one substring in a single request, then
$h\SD(w)$ atomic evaluations may require only $\SD(w)$ requests.  This gain
comes from a stronger interface, not from a smaller atomic-query complexity.

Under length-proportional weights, negative factor-OR forces
\[
 C_{\Ocal}=h\sum_{x\in\Sub^+(w)}|x|=\Theta(hn^3)
\]
on the de Bruijn-based substring-rich families used above.  Indeed, the sum is
always $O(n^3)$ because there are at most $n-L+1$ distinct factors of length
$L$.  On a linearized de Bruijn word, all factors of every length
$L\ge\ell$ are distinct.  Since $\ell=O(\log n)$, summing
$L(n-L+1)$ over $n/3\le L\le2n/3$ gives $\Omega(n^3)$.  Thus quadratically many
independent judgments can transmit cubic aggregate span length.  Prompt
sharing, prefix caching, or provider-specific pricing would require a
subadditive cost model rather than the per-key sum considered here.

\bibliographystyle{plainurl}
\bibliography{references}

@article{buhrman2002decisiontrees,
  author  = {Buhrman, Harry and de Wolf, Ronald},
  title   = {Complexity Measures and Decision Tree Complexity: A Survey},
  journal = {Theoretical Computer Science},
  year    = {2002},
  volume  = {288},
  number  = {1},
  pages   = {21--43},
  doi     = {10.1016/S0304-3975(01)00144-X}
}

@article{canonne2018adaptivity,
  author  = {Canonne, Cl{\'e}ment L. and Gur, Tom},
  title   = {An Adaptivity Hierarchy Theorem for Property Testing},
  journal = {Computational Complexity},
  year    = {2018},
  volume  = {27},
  number  = {4},
  pages   = {671--716},
  doi     = {10.1007/s00037-018-0168-4}
}

@article{chen2023semantic,
  author  = {Chen, Qiaochu and Banerjee, Arko and Demiralp, {\c{C}}a{\u{g}}atay and Durrett, Greg and Dillig, I{\c{s}}{\i}l},
  title   = {Data Extraction via Semantic Regular Expression Synthesis},
  journal = {Proceedings of the ACM on Programming Languages},
  year    = {2023},
  volume  = {7},
  number  = {OOPSLA2},
  pages   = {1848--1877},
  doi     = {10.1145/3622863}
}

@book{downey2013parameterized,
  author    = {Downey, Rodney G. and Fellows, Michael R.},
  title     = {Fundamentals of Parameterized Complexity},
  publisher = {Springer},
  address   = {London, UK},
  year      = {2013},
  doi       = {10.1007/978-1-4471-5559-1}
}

@article{huang2025membership,
  author  = {Huang, Yifei and Amini, Matin and Le Glaunec, Alexis and Mamouras, Konstantinos and Raghothaman, Mukund},
  title   = {Membership Testing for Semantic Regular Expressions},
  journal = {Proceedings of the ACM on Programming Languages},
  year    = {2025},
  volume  = {9},
  number  = {PLDI},
  pages   = {1245--1268},
  doi     = {10.1145/3729300}
}

@article{thompson1968regex,
  author  = {Thompson, Ken},
  title   = {Programming Techniques: Regular Expression Search Algorithm},
  journal = {Communications of the ACM},
  year    = {1968},
  volume  = {11},
  number  = {6},
  pages   = {419--422},
  doi     = {10.1145/363347.363387}
}

@article{antimirov1996partial,
  author  = {Antimirov, Valentin},
  title   = {Partial Derivatives of Regular Expressions and Finite Automaton Constructions},
  journal = {Theoretical Computer Science},
  year    = {1996},
  volume  = {155},
  number  = {2},
  pages   = {291--319},
  doi     = {10.1016/0304-3975(95)00182-4}
}

@inproceedings{backurs2016regex,
  author    = {Backurs, Arturs and Indyk, Piotr},
  title     = {Which Regular Expression Patterns Are Hard to Match?},
  booktitle = {2016 IEEE 57th Annual Symposium on Foundations of Computer Science},
  year      = {2016},
  pages     = {457--466},
  publisher = {IEEE},
  address   = {Los Alamitos, CA, USA},
  doi       = {10.1109/FOCS.2016.56}
}

@article{beurerkellner2023prompting,
  author  = {Beurer-Kellner, Luca and Fischer, Marc and Vechev, Martin},
  title   = {Prompting Is Programming: A Query Language for Large Language Models},
  journal = {Proceedings of the ACM on Programming Languages},
  year    = {2023},
  volume  = {7},
  number  = {PLDI},
  pages   = {1946--1969},
  doi     = {10.1145/3591300}
}

@inproceedings{bringmann2017dichotomy,
  author    = {Bringmann, Karl and Gr{\o}nlund, Allan and Larsen, Kasper Green},
  title     = {A Dichotomy for Regular Expression Membership Testing},
  booktitle = {2017 IEEE 58th Annual Symposium on Foundations of Computer Science},
  year      = {2017},
  pages     = {307--318},
  publisher = {IEEE},
  address   = {Los Alamitos, CA, USA},
  doi       = {10.1109/FOCS.2017.36}
}

@article{brzozowski1964derivatives,
  author  = {Brzozowski, Janusz A.},
  title   = {Derivatives of Regular Expressions},
  journal = {Journal of the ACM},
  year    = {1964},
  volume  = {11},
  number  = {4},
  pages   = {481--494},
  doi     = {10.1145/321239.321249}
}

@book{cygan2015parameterized,
  author    = {Cygan, Marek and Fomin, Fedor V. and Kowalik, {\L}ukasz and Lokshtanov, Daniel and Marx, D{\'a}niel and Pilipczuk, Marcin and Pilipczuk, Micha{\l} and Saurabh, Saket},
  title     = {Parameterized Algorithms},
  publisher = {Springer},
  address   = {Cham, Switzerland},
  year      = {2015},
  doi       = {10.1007/978-3-319-21275-3}
}

@inproceedings{dudek2022streaming,
  author    = {Dudek, Bart{\l}omiej and Gawrychowski, Pawe{\l} and Gourdel, Garance and Starikovskaya, Tatiana},
  title     = {Streaming Regular Expression Membership and Pattern Matching},
  booktitle = {Proceedings of the 2022 Annual ACM-SIAM Symposium on Discrete Algorithms},
  year      = {2022},
  pages     = {670--694},
  publisher = {SIAM},
  address   = {Philadelphia, PA, USA},
  doi       = {10.1137/1.9781611977073.30}
}

@book{jukna2012boolean,
  author    = {Jukna, Stasys},
  title     = {Boolean Function Complexity: Advances and Frontiers},
  publisher = {Springer},
  address   = {Berlin, Heidelberg},
  year      = {2012},
  series    = {Algorithms and Combinatorics},
  volume    = {27},
  doi       = {10.1007/978-3-642-24508-4}
}

@article{kane2013monotone,
  author  = {Kane, Daniel M.},
  title   = {A Monotone Function Given by a Low-Depth Decision Tree That Is Not an Approximate Junta},
  journal = {Theory of Computing},
  year    = {2013},
  volume  = {9},
  number  = {17},
  pages   = {587--592},
  doi     = {10.4086/toc.2013.v009a017}
}

@inproceedings{kupferman2002extended,
  author    = {Kupferman, Orna and Zuhovitzky, Sharon},
  title     = {An Improved Algorithm for the Membership Problem for Extended Regular Expressions},
  booktitle = {Mathematical Foundations of Computer Science 2002},
  year      = {2002},
  pages     = {446--458},
  publisher = {Springer},
  address   = {Berlin, Heidelberg},
  doi       = {10.1007/3-540-45687-2_37}
}

@article{kuchnik2023relm,
  author  = {Kuchnik, Michael and Smith, Virginia and Amvrosiadis, George},
  title   = {Validating Large Language Models with {ReLM}},
  journal = {Proceedings of Machine Learning and Systems},
  year    = {2023},
  volume  = {5},
  pages   = {457--476}
}

@inproceedings{green2007provenance,
  author    = {Green, Todd J. and Karvounarakis, Grigoris and Tannen, Val},
  title     = {Provenance Semirings},
  booktitle = {Proceedings of the Twenty-Sixth ACM SIGMOD-SIGACT-SIGART Symposium on Principles of Database Systems},
  year      = {2007},
  pages     = {31--40},
  publisher = {ACM},
  address   = {New York, NY, USA},
  doi       = {10.1145/1265530.1265535}
}

@article{deshpande2016stochastic,
  author  = {Deshpande, Amol and Hellerstein, Lisa and Kletenik, Devorah},
  title   = {Approximation Algorithms for Stochastic Submodular Set Cover with Applications to Boolean Function Evaluation and Min-Knapsack},
  journal = {ACM Transactions on Algorithms},
  year    = {2016},
  volume  = {12},
  number  = {3},
  pages   = {1--28},
  doi     = {10.1145/2876506}
}

@article{nisan1991crew,
  author  = {Nisan, Noam},
  title   = {{CREW PRAMs} and Decision Trees},
  journal = {SIAM Journal on Computing},
  year    = {1991},
  volume  = {20},
  number  = {6},
  pages   = {999--1007},
  doi     = {10.1137/0220062}
}

@inproceedings{petersen2002intersection,
  author    = {Petersen, Holger},
  title     = {The Membership Problem for Regular Expressions with Intersection Is Complete in {LOGCFL}},
  booktitle = {STACS 2002},
  year      = {2002},
  pages     = {513--522},
  publisher = {Springer},
  address   = {Berlin, Heidelberg},
  doi       = {10.1007/3-540-45841-7_42}
}

@inproceedings{rosu2007extended,
  author    = {Ro{\c{s}}u, Grigore},
  title     = {An Effective Algorithm for the Membership Problem for Extended Regular Expressions},
  booktitle = {Foundations of Software Science and Computation Structures},
  year      = {2007},
  pages     = {332--345},
  publisher = {Springer},
  address   = {Berlin, Heidelberg},
  doi       = {10.1007/978-3-540-71389-0_24}
}

@inproceedings{saks1986probabilistic,
  author    = {Saks, Michael and Wigderson, Avi},
  title     = {Probabilistic Boolean Decision Trees and the Complexity of Evaluating Game Trees},
  booktitle = {27th Annual Symposium on Foundations of Computer Science},
  year      = {1986},
  pages     = {29--38},
  publisher = {IEEE},
  address   = {Los Alamitos, CA, USA},
  doi       = {10.1109/SFCS.1986.44}
}

@techreport{vazquezchanlatte2024lstarlm,
  author      = {Vazquez-Chanlatte, Marcell and Elmaaroufi, Karim and Witwicki, Stefan and Seshia, Sanjit A.},
  title       = {$L^*LM$: Learning Automata from Examples Using Natural Language Oracles},
  institution = {arXiv},
  year        = {2024},
  number      = {2402.07051},
  doi         = {10.48550/arXiv.2402.07051}
}

@article{varatalu2025resharp,
  author    = {Varatalu, Ian Erik and Veanes, Margus and Ernits, Juhan},
  title     = {{RE\#}: High Performance Derivative-Based Regex Matching with Intersection, Complement, and Restricted Lookarounds},
  journal   = {Proceedings of the ACM on Programming Languages},
  year      = {2025},
  volume    = {9},
  number    = {POPL},
  articleno = {1},
  numpages  = {32},
  doi       = {10.1145/3704837}
}

@inproceedings{wang2019hyperscan,
  author    = {Wang, Xiang and Hong, Yang and Chang, Harry and Park, KyoungSoo and Langdale, Geoff and Hu, Jiayu and Zhu, Heqing},
  title     = {Hyperscan: A Fast Multi-Pattern Regex Matcher for Modern {CPUs}},
  booktitle = {16th USENIX Symposium on Networked Systems Design and Implementation},
  year      = {2019},
  pages     = {631--648},
  publisher = {USENIX Association},
  address   = {Berkeley, CA, USA},
  url       = {https://www.usenix.org/conference/nsdi19/presentation/wang-xiang}
}

\end{document}